\documentclass[conference]{IEEEtran}

\ifCLASSINFOpdf
\else
\fi

\usepackage{amsfonts}
\usepackage{graphicx,subfig}
\usepackage{amsmath,amsthm}
\usepackage{wrapfig}

\usepackage{epsfig}
\usepackage{amsmath}
\usepackage{amssymb}
\usepackage{psfrag}
\usepackage[absolute]{textpos}

\newtheorem{thm}{Theorem}
\newtheorem{lem}{Lemma}
\newtheorem{definition}{Definition}

\newtheorem{algorithm}{Algorithm}

\newcommand{\A}{{\bf A}}
\def\x{{\bf x}}
\def\y{{\bf y}}
\def\w{{\bf w}}

\def\z{{\bf z}}
\def\e{{\bf e}}

\newcommand{\beq}{\begin{equation}}
\newcommand{\eeq}{\end{equation}}
\newcommand{\bea}{\begin{eqnarray}}
\newcommand{\eea}{\end{eqnarray}}

\newcommand{\Prob}{\ensuremath{\mathbb{P}}}

\long\def\symbolfootnote[#1]#2{\begingroup%
\def\thefootnote{\fnsymbol{footnote}}\footnote[#1]{#2}\endgroup}

\begin{document}
%
\title{Improved Sparse Recovery Thresholds with Two-Step Reweighted $\ell_1$ Minimization}


\author{
\authorblockN{M.~Amin Khajehnejad}
\authorblockA{Caltech\\
Pasadena CA 91125, USA \\
Email: amin@caltech.edu} \and
\authorblockN{Weiyu Xu}
\authorblockA{Cornell University \\
 Ithaca NY 14853, USA \\
Email: weiyu@caltech.edu} \and
\authorblockN{A.~Salman Avestimehr}
\authorblockA{Cornell University \\
Ithaca NY 14853, USA \\
Email: avestimehr@ece.cornell.edu} \and
\authorblockN{Babak Hassibi}
\authorblockA{Caltech\\
Pasadena CA 91125, USA \\
Email: hassibi@caltech.edu}}


%


\maketitle

\begin{abstract}
\boldmath
It is well known that $\ell_1$ minimization can be used to recover sufficiently sparse unknown signals from compressed linear measurements. In fact, exact thresholds on the sparsity, as a function of the ratio between
the system dimensions,
so that with high probability almost all sparse signals can be
recovered from iid Gaussian measurements, have been computed and are
referred to as ``weak thresholds'' \cite{D}. In this paper, we
introduce a reweighted $\ell_1$ recovery algorithm composed of two
steps: a standard $\ell_1$ minimization step to identify a set of
entries where the
signal is likely to reside, and a weighted $\ell_1$ minimization step where entries outside this set are penalized. For signals
where the non-sparse component has iid Gaussian entries, we prove a
``strict'' improvement in the weak recovery threshold. Simulations
suggest that the improvement can be quite impressive---over 20\%  in
the example we consider.
\end{abstract}


%
\IEEEpeerreviewmaketitle

\section{Introduction}
\label{sec:Intro}
Compressed sensing addresses the problem of recovering sparse signals
from under-determined systems of linear equations  \cite{Rice}. In
particular, if $\x$ is a $n\times1$ real vector that is known to have at
most $k$ nonzero elements where $k<n$, and $\A$ is a $m\times n $
measurement matrix with $k<m<n$, then for appropriate values of $k$,
$m$ and $n$, it is possible to efficiently recover $\x$ from $\y=\A\x$ \cite{D
  CS,DT,CT,Baraniuk Manifolds}. The most well recognized such
algorithm is $\ell_1$ minimization which can be formulated  as
follows:
\beq
\label{eq:l1 min}
\min_{\A\z=\A\x}\|\z\|_1
\eeq
The first result that established the fundamental limits of signal
recovery using $\ell_1$ minimization is due to Donoho and Tanner
\cite{D,DT}, where it is shown
that if the measurement matrix is iid Gaussian, for a given ratio
of $\delta = \frac{m}{n}$, $\ell_1$ minimization can successfully
recover {\em every} $k$-sparse signal, provided that $\mu = \frac{k}{n}$ is
smaller that a certain threshold. This statement is true
asymptotically as $n\rightarrow \infty$ and with high
probability. This threshold guarantees the recovery of {\em all}
sufficiently sparse signals and is therefore referred to as a strong
threshold. It therefore does not depend on the actual distribution of
the nonzero entries of the sparse signal and as such is a universal
result. At this point it is not known whether there exists other
polynomial-time algorithms with superior strong threshold.

Another notion introduced and computed in \cite{D,DT} is that of a
{\em weak} threshold where signal recovery is guaranteed for {\em
  almost all} support sets and {\em almost all} sign patterns of the
sparse signal, with high probability as $n\rightarrow\infty$. The weak
threshold is the one that can be observed in simulations of $\ell_1$
minimization and allows for signal recovery beyond the strong threshold. It
is also universal in the sense that it applies to all symmetric
distributions that one may draw the nonzero signal entries
from. Finally, it is not known whether there exists other
polynomial-time algorithms with superior weak thresholds.

In this paper we prove that a certain \emph{iterative
  reweighted $\ell_1$} algorithm indeed has better weak recovery
guarantees for particular classes of sparse signals, including sparse
Gaussian signals.  We had previously introduced these algorithms in
\cite{Allerton reweighted l_1}, and had proven that for a very
restricted class of \emph{polynomially decaying} sparse signals they
outperform standard $\ell_1$ minimization. In this paper however, we
  extend this result to a much wider and more reasonable class of
sparse signals. The key to our result is the fact that for these
  classes of signals, $\ell_1$ minimization has an \emph{approximate support
  recovery} property which can be exploited via a reweighted
  $\ell_1$ algorithm, to obtain a provably superior weak threshold.
In particular, we consider Gaussian sparse
signals, namely sparse signals where the nonzero entries are iid
  Gaussian. Our analysis of Gaussian sparse signals relies on
  concentration bounds on the partial sum of their order
  statistics. Though not done here, it can be shown that
for symmetric distributions with sufficiently fast decaying tails and
nonzero value at the origin, similar bounds and improvements on the
  weak threshold can be achieved.

It is worth noting that different variations of reweighted $\ell_1$
algorithms have been recently introduced in the literature and, have
shown experimental improvement over ordinary $\ell_1$ minimization
\cite{CWB07,Needell}.  In \cite{Needell} approximately sparse signals
have been considered, where perfect recovery is never
possible. However, it has been shown that the recovery noise can be
reduced using an iterative scheme. In \cite{CWB07}, a similar
algorithm is suggested and is empirically shown to outperform $\ell_1$
minimization for exactly sparse signals with non-flat
distributions. Unfortunately, \cite{CWB07} provides no theoretical
analysis or performance guarantee. The particular reweighted $\ell_1$
minimization algorithm that we propose and analyze is of signiciantly
less computational complexity than the earlier ones (it only solves
two linear programs). Furthermore, experimental results confirm that
it exhibits much better performance than previous reweighted methods.
Finally, while we do rigorously establish an
{\em improvement} in the weak threshold, we currently do not have
tight bounds on the new weak threshold and simulation results are far
better than the bounds we can provide at this time.

\hfill 
\section{Basic Definitions}
 A sparse signal with exactly $k$ nonzero entries is called
 $k$-sparse. For a vector $\x$, $\|\x\|_1$ denotes the $\ell_1$
 norm. The support (set) of $\x$,  denoted by $supp(\x)$, is the index
 set of its nonzero coordinates. For a vector $\x$ that is not exactly
 $k$-sparse, we define the $k$-support of $\x$ to be the index set of
 the largest $k$ entries of $\x$ in amplitude, and denote it by
 $supp_k(\x)$. For a subset $K$ of the entries of $\x$, $\x_K$ means
 the vector formed by those entries of $\x$ indexed in $K$. Finally,
 $\max|\x|$ and $\min|\x|$ mean the absolute value of the maximum and
 minimum entry of $\x$ in magnitude, respectively.

\section{Signal Model and Problem Description}
\label{sec:model}
We consider sparse random signals with iid Gaussian nonzero
entries. In other words we assume that the unknown sparse signal is a
$n\times 1$ vector $\x$ with exactly $k$ nonzero entries, where each
nonzero entry is independently derived from the standard normal distribution
$\mathcal{N}(0,1)$. The measurement matrix $\A$ is a $m\times n$
matrix with iid Gaussian entries with a ratio of dimensions $\delta
= \frac{m}{n}$. Compressed sensing theory guarantees that if
$\mu=\frac{k}{n}$ is smaller than a certain threshold, then every
$k$-sparse signal can be recovered using $\ell_1$ minimization. The
relationship between $\delta$ and the maximum threshold of $\mu$ for
which such a guarantee exists is called the \emph{strong sparsity
  threshold}, and is denoted by $\mu_{S}(\delta)$.  A more practical
performance guarantee is the so-called \emph{weak sparsity threshold},
denoted by $\mu_{W}(\delta)$, and has the following
interpretation. For a fixed value of $\delta = \frac{m}{n}$ and
iid Gaussian matrix $\A$ of size $m\times n$,  a random $k$-sparse
vector $\x$ of size $n\times 1$ with a randomly chosen support set and
a random sign pattern can be recovered from $\A\x$ using $\ell_1$
minimization with high probability, if
$\frac{k}{n}<\mu_{W}(\delta)$. Similar recovery thresholds can be
obtained by imposing more or less restrictions. For example, strong
and weak thresholds for nonnegative signals have been evaluated in
\cite{Donoho positive}.

We assume that the support size of $\x$, namely $k$, is slightly
larger than the weak threshold of $\ell_1$ minimization. In other
words,  $k = (1+\epsilon_0)\mu_{W}(\delta)$ for some
$\epsilon_0>0$. This means that if we use $\ell_1$ minimization, a
randomly chosen $\mu_{W}(\delta)n$-sparse signal will be recovered
perfectly with very high probability, whereas a randomly selected
$k$-sparse signal will not. We would like to show that for a strictly
positive $\epsilon_0$, the iterative reweighted $\ell_1$ algorithm of
Section \ref{sec:Algorithm} can indeed recover a randomly selected
$k$-sparse signal with high probability, which means that it has an
improved weak threshold.

\section{Iterative weighted $\ell_1$ Algorithm}
\label{sec:Algorithm}
We propose the following algorithm, consisting of two $\ell_1$ minimization steps: a standard one and a weighted one. The input to
the algorithm is the vector $\y=\A\x$, where $\x$ is a $k$-sparse
signal with $k=(1+\epsilon_0)\mu_W(\delta)n$, and the output is an
approximation $\x^*$ to the unknown vector $\x$. We assume that $k$,
or an upper bound on it, is known. Also $\omega>1$ is a
predetermined weight.

\begin{figure}[t]
\centering
  \includegraphics[width= 0.27\textwidth]{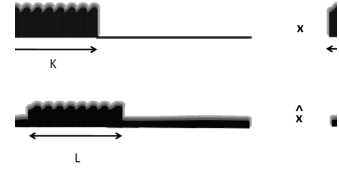}
  \caption{ \scriptsize{A pictorial example of a sparse signal  and its $\ell_1$ minimization approximation}}
  \label{fig:sigmal}
\end{figure}
\begin{algorithm}
\text{}
\begin{enumerate}
\item Solve the $\ell_1$ minimization problem:
\begin{equation}
\hat{\x} = \arg{ \min{ \|\z\|_1}}~~\text{subject to}~~ \A\z
= \A\x.
\end{equation}
\item Obtain an approximation for the support set of $\x$:
find the index set $L \subset \{1,2, ..., n\}$ which corresponds to
the largest $k$ elements of
$\hat{\x}$ in magnitude.
\item Solve the following weighted $\ell_1$ minimization problem and declare the solution as output:
 \beq
\x^* = \arg{\min\|\z_L\|_1+\omega\|\z_{\overline{L}}\|_1}~~\text{subject to}~~ \A\z
= \A\x.
\label{eq:weighted l_1}
\eeq
\end{enumerate}
\label{alg:modmain}
\end{algorithm}

The intuition behind the algorithm should be clear. In the first step
we perform a standard $\ell_1$ minimization. If the sparsity of the
signal is beyond the weak threshold $\mu_W(\delta)n$, then $\ell_1$ minimization is not capable of recovering the signal. However, we use
the output of the $\ell_1$ minimization to identify an index set, $L$,
which we ``hope'' contains most of the nonzero entries of $\x$. We
finally perform a weighted $\ell_1$ minimization by penalizing those
entries of $\x$ that are not in $L$ (ostensibly because they have a
lower chance of being nonzero).

In the next sections we formally prove that the above intuition is
correct and that, for certain classes of signals, Algorithm
\ref{alg:modmain} has a recovery threshold beyond that of standard
$\ell_1$ minimization. The idea of the proof is as follows. In section
\ref{sec:robustness}, we prove that there is a large overlap between
the index set $L$, found in step 2 of the algorithm, and the support set
of the unknown signal $\x$ (denoted by $K$)---see Theorem \ref{thm:l_1
  support recovery} and Figure \ref{fig:sigmal}. Then in section
\ref{sec:perfect recovery}, we show that the large overlap between $K$
and $L$ can result in perfect recovery of $\x$, beyond the standard
weak threshold, when a weighted $\ell_1$ minimization is used in step
3.

\section{Approximate Support Recovery, Steps 1 and 2 of the Algorithm}
\label{sec:robustness}

In this Section, we carefully study the first two steps of Algorithm
\ref{alg:modmain}. The unknown signal $\x$ is assumed to be a Gaussian
$k$-sparse vector with support set $K$, where
$k=|K|=(1+\epsilon_0)\mu_{W}(\delta)n$, for some $\epsilon_0>0$. By A
Gaussian $k$-sparse vector, we mean one where the nonzero entries are
iid Gaussian (zero mean and unit variance, say). The
solution $\hat{\x}$ to the
$\ell_1$ minimization obtained in step 1 of Algorithm
\ref{alg:modmain} is in all likelihood a full vector.  The set $L$, as
defined in the algorithm, is in fact the
$k$-support set of $\hat{\x}$. We show that for small enough
$\epsilon_0$, the intersection of  $L$ and $K$ is with high
probability very large, so that $L$ can be counted as a good
approximation to $K$ (Figure \ref{fig:sigmal}).

In order to find a decent lower bound on $|L\cap K|$, we mention three
separate facts and establish a connection between them. First, we
prove a general lemma that bounds $|K\cap L|$ as a function of
$\|\x-\hat{\x}\|_1$. Then, we mention an intrinsic property of
$\ell_1$ minimization called \emph{weak robustness} that provides an
upper bound on the quantity $\|\x-\hat{\x}\|_1$. Finally, we
specifically use the Gaussianity of $\x$ to obtain Theorem
\ref{thm:l_1 support recovery}.  Let us start with a definition.
\begin{definition}
For a $k$-sparse signal $\x$, we define $W(\x,\lambda)$ to be the size of the largest subset of nonzero entries of $\x$ that has a $\ell_1$ norm less than or equal to $\lambda$.
\beq
\nonumber W(\x,\lambda):= \max\{|S|~|~S\subseteq supp(\x),~\|\x_S\|_1\leq \lambda\}
\eeq
\end{definition}
\noindent Note that $W(\x,\lambda)$ is increasing in $\lambda$.
\begin{lem}
Let $\x$ be a $k$-sparse vector and $\hat{\x}$ be another vector. Also, let $K$ be the support set of $\x$ and $L$ be the $k$-support set of $\hat{\x}$. Then
\beq
|K\cap L|\geq k-W(\x,\|\x-\hat{\x}\|_1)
\eeq
\label{lem:deviation thm}
\end{lem}
\begin{proof}
Let $x_i$ be the $i$th entry of $\x$ and  $\e^*$ be the solution to the following minimization program
\begin{align}
&\min\|\e\|_1\nonumber \\
\text{s.t}. &\left\{\begin{array}{c}\max|(\x+\e)_{K\setminus L}| \leq \min|(\x+\e)_{K\cap L}| \\ \max|(\x+\e)_{K\setminus L}| \leq \min|(\x+\e)_{L\setminus K}|\end{array}\right.
\label{eq:min aux}
\end{align}
\noindent Now since, $ \hat{\x} = \x+(\hat{\x}-\x)$ satisfies the
constraints of the minimization  (\ref{eq:min aux}), we can write
\beq
\|\e^*\|_1 \leq \|\hat{\x}-\x\|_1.
\label{eq:aux1}
\eeq
\noindent Let $a = \max|(\x+\e^*)_{K\setminus L}|$.  Then for each
$i\in K\setminus L$, using the triangular inequality we have
\beq
|x_i|-|e_i| \leq |x_i + e_i| \leq a
\label{eq:aux-+}
\eeq
Therefore, by summing up the inequalities in (\ref{eq:aux-+}) for $i\in K\setminus L$ we have
\beq
\|\e^*_{K\setminus L}\|_1 \geq \sum_{i\in K\setminus L, |x_i| >a}|x_i|-a
\eeq
Similarly,
\beq
\|\e^*_{L\setminus K}\|_1 \geq a|L\setminus K|
\eeq
But $|L\setminus K|=|K\setminus L|$ and therefore we have
\bea
\|\e^*\|_1 &\geq&  \sum_{i\in K\setminus L, |x_i| >a}(|x_i|-a) + a|K\setminus L|  \nonumber \\
&\geq& \sum_{i\in K\setminus L}|x_i|
\label{eq:aux2}
\eea
\noindent (\ref{eq:aux1}) and (\ref{eq:aux2}) together imply that $\|\x-\hat{\x}\|_1\geq \|\x_{K\setminus L}\|_1$, which by definition means that $W(\x,\|\x-\hat{\x}\|_1)\geq |K\setminus L|$.
\end{proof}
We now introduce the notion of weak robustness, which
allows us to bound $\|\x-\hat{\x}\|_1$, and has the following formal
definition \cite{isitrobust}.
\begin{definition}
Let the set $S\subset\{1,2,\cdots,n\}$ and the subvector $\x_S$  be
fixed. A solution $\hat{\x}$ is called weakly robust if, for some $C >
1$ called the robustness factor, and all $\x_{\overline{S}}$, it
holds that
\beq
\|(\x-\hat{\x})_S\|_1 \leq \frac{2C}{C-1}\|\x_{\overline{S}}\|_1
\label{def1}
\eeq
\noindent and
\beq
\|\x_S\|-\|\hat{\x}_S\|\leq \frac{2}{C-1}\|\x_{\overline{S}}\|_1
\label{def2}
\eeq
\end{definition}
\noindent The weak robustness notion allows us to bound the error
in $\|\x -\hat{\x}\|_1$ in the following way. If the matrix $\A_S$ ,
obtained by retaining only those columns of $\A$ that are
indexed by $S$, has full column rank, then the quantity
\beq
\nonumber
~\kappa = \max_{\A\w =0, \w\neq 0} \frac{\|\w_S\|_1}{\|\w_{\overline{S}}\|_1}
\eeq
must be finite, and one can write
\beq
\|\x-\hat{\x}\|_1 \leq \frac{2C(1+\kappa)}{C-1}\|\x_{\overline{S}}\|_1
\label{eq:robustness}
\eeq
\noindent In \cite{isitrobust}, it has been shown that for Gaussian
iid measurement matrices $\A$, $\ell_1$ minimization is weakly robust,
i.e., there exists a robustness factor
$C>1$ as a function of $\frac{|S|}{n} < \mu_{W}(\delta)$ for which
(\ref{def1}) and (\ref{def2})
hold. Now let $k_1 = (1-\epsilon_1)\mu_{W}(\delta)n$ for some small
$\epsilon_1>0$, and $K_1$ be the $k_1$-support set of $\x$, namely,
the set of the largest $k_1$ entries of $\x$ in magnitude. Based on
equation (\ref{eq:robustness}) we may write
\beq
\|\x-\hat{\x}\|_1 \leq \frac{2C(1+\kappa)}{C-1}\|\x_{\overline{K_1}}\|_1
\label{eq:robustness1}
\eeq
\noindent For a fixed value of $\delta$, $C$ in (\ref{eq:robustness1})
is a function of $\epsilon_1$ and becomes arbitrarily close to $1$ as
$\epsilon_1\rightarrow 0$. $\kappa$ is also a bounded function of
$\epsilon_1$ and therefore we may replace it with an upper bound
$\kappa^*$. We now have a bound on $\|\x-\hat{\x}\|_1$. To explore
this inequality and understand its asymptotic behavior, we
apply a third result, which is a certain concentration bound on
the order statistics of Gaussian random variables.
\begin{lem}
Suppose $X_1,X_2,\cdots,X_n$ are $N$ iid $\mathcal{N}(0,1)$ random
variables. Let $S_N = \sum_{i=1}^{N}|X_i|$ and let $S_M$ be the sum of
the largest $M$ numbers among the $|X_i|$, for each
$1\leq M < N$. Then for every $\epsilon>0$, as $N\rightarrow\infty$,
we have
\bea
\Prob(|\frac{S_N}{N} - \sqrt{\frac{2}{\pi}}| > \epsilon)\rightarrow 0, \\
\Prob(|\frac{S_M}{S_N} -
\exp(-\frac{\Psi^2(\frac{M}{2N})}{2})|>\epsilon )\rightarrow 0
\eea
\noindent where $\Psi(x) = Q^{-1}(x)$ with $Q(x) = \frac{1}{\sqrt{2\pi}}\int_{x}^{\infty}e^{-\frac{y^2}{2}}dy$.
\label{lemma:Gaussian_Base}
\end{lem}
\noindent As a direct consequence of Lemma \ref{lemma:Gaussian_Base} we can write:
\beq
\Prob(|\frac{\|\x_{\overline{K_1}}\|_1}{\|\x\|_1} - (1-e^{-0.5\Psi^2(0.5\frac{1-\epsilon_1}{1+\epsilon_0})})|>\epsilon)\rightarrow0
\label{eq:kbar bound}
\eeq
\noindent for all $\epsilon>0$ as $n\rightarrow\infty$. Define
\beq
\nonumber \zeta(\epsilon_0):= \inf_{\epsilon_1>0}\frac{2C(1+\kappa^*)}{C-1}(1-e^{-0.5\Psi^2(0.5\frac{1-\epsilon_1}{1+\epsilon_0})})
\eeq

\noindent Incorporating (\ref{eq:robustness1}) into (\ref{eq:kbar bound}) we may write
\beq
\Prob(\frac{\|\x-\hat{\x}\|_1}{\|\x\|_1} - \zeta(\epsilon_0) < \epsilon)\rightarrow 1
\label{eq:robustness2}
\eeq
\noindent for all $\epsilon>0$ as $n\rightarrow\infty$. Let us
summarize our conclusions so far. First, we were able to show that
$|K\cap L|\geq k-W(\x,\|\x-\hat{\x}\|_1)$. Weak
robustness of $\ell_1$ minimization and Gaussianity of the signal then
led
us to the fact that for large $n$ with high probability $\|\x-\hat{\x}\|_1
\leq \zeta(\epsilon_0)\|\x\|_1$. These results build up the next key
theorem, which is the conclusion of this section.
\begin{thm}[Support Recovery]
Let $\A$ be an iid Gaussian $m\times n$ measurement matrix with
$\frac{m}{n}=\delta$. Let $k=(1+\epsilon_0)\mu_{W}(\delta)$ and $\x$
be a $n\times1$ random Gaussian $k$-sparse  signal. Suppose that
$\hat{\x}$ is the approximation to $\x$ given by the $\ell_1$ minimization, i.e. $\hat{\x}=argmin_{\A\z=\A\x}\|\z\|_1$. Then, as
$n\rightarrow\infty$, for all $\epsilon>0$,
\beq
\small
\Prob(\frac{|supp(\x) \cap supp_k(\hat{\x})|}{k} -
2Q(\sqrt{-2\log(1-\zeta(\epsilon_0))})>-\epsilon)\rightarrow 1.
\label{eq:support recovery}
\eeq
\label{thm:l_1 support recovery}
\end{thm}
\begin{proof}
For each $\epsilon'>0$ and large enough $n$, with high probability it
holds that
$\|\x-\hat{\x}\|_1<(\zeta(\epsilon_0)+\epsilon')\|\x\|_1$. Therefore,
from Lemma \ref{lem:deviation thm} and the fact that $W(\x,\lambda)$
is increasing in $\lambda$, $|K\cap L| \geq k -
W(\x,(\zeta(\epsilon_0)+\epsilon')\|\x\|_1)$ with high
probability. Also, an implication of Lemma \ref{lemma:Gaussian_Base}
reveals that for any positive $\epsilon''$ and $\alpha$, $\frac{W(\x,\alpha\|x\|_1)}{k}<(1-2Q(\sqrt{-2\log(1-\alpha)}))+\epsilon''$
for large enough $n$. Putting these together, we conclude that with very high
probability $\frac{|K\cap L|}{k}\geq
2Q(\sqrt{-2\log(1-\zeta(\epsilon_0)-\epsilon')})-\epsilon''$. The
desired result now follows from the continuity of the $\log(\cdot)$
and $Q(\cdot)$ functions.
\end{proof}
\noindent Note that if $\lim_{\epsilon_0\rightarrow0}\zeta(\epsilon_0)=0$, then  Theorem \ref{thm:l_1 support recovery} implies that $\frac{|K\cap L|}{k}$ becomes arbitrarily close to 1.

 \section{Perfect Recovery, Step 3 of the Algorithm}
 \label{sec:perfect recovery}
 In Section \ref{sec:robustness} we showed that. if
 $\epsilon_0$ is small, the $k$-support of $\hat{\x}$, namely
 $L=supp_k(\hat{\x})$, has a significant overlap with the true support of
 $\x$. We even found a quantitative lower bound on the size of this overlap
 in Theorem \ref{thm:l_1 support recovery}. In step 3 of
 Algorithm \ref{alg:modmain}, weighted $\ell_1$ minimization is used,
 where the entries in  $\overline{L}$ are assigned a higher weight
 than those in $L$. In \cite{isitweighted}, we have been able to
 analyze the performance of such weighted $\ell_1$ minimization
 algorithms. The idea is that if a sparse vector $\x$ can
 be partitioned into two sets $L$ and $\overline{L}$, where in one set
 the fraction of non-zeros is much larger than in the other set, then
 (\ref{eq:weighted l_1}) can potentially recover $\x$ with an
 appropriate choice of the weight $\omega > 1$, even though
 $\ell_1$ minimization cannot.  The following theorem can be deduced
 from \cite{isitweighted}.
\begin{thm}
Let $L\subset \{1,2,\cdots,n\}$ , $\omega>1$ and the fractions
$f_1,f_2\in[0,1]$ be given. Let $\gamma_1 = \frac{|L|}{n}$ and
$\gamma_2=1-\gamma_1$. There exists a threshold
$\delta_c(\gamma_1,\gamma_2,f_1,f_2,\omega)$ such that with high
probability, almost all random sparse vectors $\x$ with \emph{at least}
$f_1\gamma_1n$ nonzero entries over the set $L$, and \emph{at most}
$f_2\gamma_2n$ nonzero entries over the set $\overline{L}$ can be
perfectly recovered using
$\min_{\A\z=\A\x}\|\z_L\|_1+\omega\|\z_{\overline{L}}\|_1$, where $\A$
is a $\delta_cn\times n$ matrix with iid Gaussian entries.

\noindent Furthermore, for appropriate $\omega$,
\[
\mu_W(\delta_c(\gamma_1,\gamma_2,f_1,f_2,\omega))<f_1\gamma_1+f_2\gamma_2,
\]
i.e., standard $\ell_1$ minimization using a $\delta_cn\times n$
measurement matrix with iid Gaussian entries cannot recover such $x$.
\label{thm:delta}
\end{thm}
\noindent For completeness, in Appendix \ref{App:delta_c}, we provide
the calculation of $\delta_c(\gamma_1,\gamma_2,f_1,f_2,\omega)$. A
software package for computing such thresholds can also be found in
\cite{sotware link}.

\begin{thm}[Perfect Recovery]
Let $\A$ be a $m\times n$ i.i.d. Gaussian matrix with
$\frac{m}{n}=\delta$. If
$\lim_{\epsilon_0\rightarrow0}\zeta(\epsilon_0)=0$ and
$\delta_c(\mu_{W}(\delta),1-\mu_{W}(\delta),1,0,\omega) < \delta$,
then there exist $\epsilon_0>0$ and $\omega>0$ so that Algorithm
\ref{alg:modmain} perfectly recovers a random
$(1+\epsilon_0)\mu_{W}(\delta)$-sparse vector with i.i.d. Gaussian
entries with high probability as $n$ grows to infinity.
\label{thm: final thm}
\end{thm}
\begin{proof}
Leveraging on the statement of Theorem \ref{thm:delta}, in order to
show that $\x$ is perfectly recovered in the last step of the
algorithm, , it is sufficient to find the overlap fractions
$f_1=\frac{|L\cap K|}{|L|}$ and $f_2=\frac{|\overline{L}\cap
  K|}{|\overline{L}|}$ for a given $\epsilon_0$, and show that
$\delta_c(\frac{k}{n},1-\frac{k}{n},f_1,f_2,\omega) \leq \delta$.  On
the other hand, according to Theorem \ref{thm:l_1 support recovery}
as $\epsilon_0\rightarrow0$, $f_1\rightarrow1$ and
$f_2\rightarrow0$. Therefore, if
$\delta_c(\mu_{W}(\delta),1-\mu_{W}(\delta),1,0,\omega) < \delta$,
from the continuity of $\delta_c$ we can conclude that for a strictly
positive $\epsilon_0$ and corresponding overlap fractions $f_1$ and
$f_2$,
$\delta_c((1+\epsilon_0)\mu_{W}(\delta),1-(1+\epsilon_0)\mu_{W}
(\delta),f_1,f_2,\omega)< \delta$, which completes the proof.
\end{proof}
For $\delta= 0.555$ it is  easy to verify numerically that the
conditions of Theorem \ref{thm: final thm} hold. We haven  chosen
$\alpha=1$ and  have computed an approximate upper bound
$\zeta^*(\epsilon_0)$ for $\zeta(\epsilon_0)$, using the results of \cite{isitrobust}. This is
depicted in Figure \ref{fig:zeta(e0)}. As shown, when
$\epsilon_0\rightarrow0$, $\zeta^*(\epsilon_0)$ becomes arbitrarily
small too. Using this curve and the numerical $\delta_c$ function from
Appendix \ref{App:delta_c}, we can show that for $\omega=10$, the
value of $\epsilon_0=5\times10^{-4}$ satisfies the statement of
Theorem \ref{thm: final thm}.  This improvement is of course much
smaller than what we observe in practice.
\begin{figure}[t]
\centering
  \includegraphics[width= 0.4\textwidth]{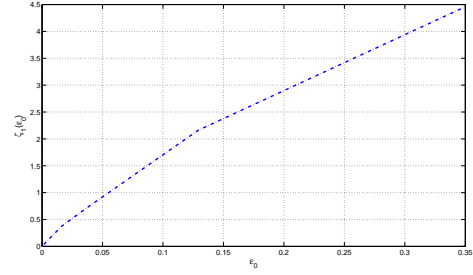}
  \caption{ \scriptsize{ An approximate upper bound for $\zeta(\epsilon_0)$ for $\delta=0.555$ .}}
  \label{fig:zeta(e0)}
\end{figure}

\section{Beyond Gaussians and Simulations}
\label{sec:simulation}

It is reasonable to ask if we can prove a theoretical threshold
improvement for sparse signals with other distributions. The attentive
reader will note that the only step where we used the Gaussianity of
the signal was in the the order statistics results of Lemma
\ref{lemma:Gaussian_Base}. This result has the following
interpretation. For $N$ iid random variables, the
ratio $\frac{S_M}{S_N}$ can be
approximated by a known function of $\frac{M}{N}$. In the Gaussian
case, this function behaves as $(1-\frac{M}{N})^2$, as $M\rightarrow
N$. For constant magnitude signals (say BPSK), the function behaves as
$1-\frac{M}{N}$, for $M\rightarrow N$, which proves that the
reweighted method yields no improvement. A more careful analysis,
beyond the scope and space of this paper, reveals that the improvement
over $\ell_1$ minimization depends on the behavior of
$\frac{S_M}{S_N}$, as $M\rightarrow N$, which in term depends on the
smallest order $n$ for which $f^{(n)}(0)\neq 0$, i.e., the smallest
$n$ such that the $n$-th
derivative of the distribution at the origin is nonzero.

\begin{figure}[t]
\centering
  \includegraphics[width= 0.5\textwidth]{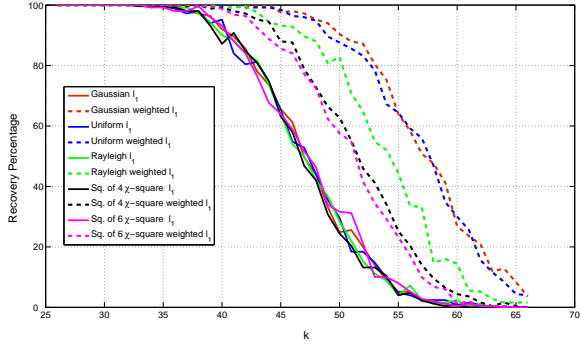}
  \caption{ \scriptsize{Empirical Recovery Percentage for $n=200$ and $\delta = 0.5555$.}}
  \label{fig:simultions}
\end{figure}

These are exemplified by the simulations in Figure
\ref{fig:simultions}. Here the signal
dimension is $n=200$, and the number of measurements is $m=112$, which
corresponds to a value of $\delta = 0.5555$. We generated random
sparse signals with iid entries coming from certain
distributions; Gaussian, uniform,  Rayleigh , square root of
$\chi$-square with 4 degrees of freedom and, square root of
$\chi$-square with 6 degrees of freedom. Solid lines represent the simulation
results for ordinary $\ell_1$ minimization, and different colors
indicate different distributions. Dashed lines are used to show the
results for Algorithm
\ref{alg:modmain}. Note that the more derivatives that vanish at the
origin, the less the improvement over $\ell_1$ minmimization. The
Gaussian and uniform distributions are flat
and nonzero at the origin and show an impressive more than 20\%
improvement in the weak threshold (from 45 to 55).



\begin{thebibliography}{1}


\bibitem{D CS}  D. Donoho,`` Compressed sensing'', {\it IEEE Trans. on Information Theory}, 52(4), pp. 1289 - 1306, April 2006)

\bibitem{DT} D. Donoho,
``High-Dimensional Centrally Symmetric Polytopes with Neighborliness
Proportional to Dimension '', {\it Discrete and Computational
Geometry }, 102(27), pp. 617-652, 2006, Springer .


\bibitem{CT} E. Cand\'{e}s and T. Tao,
``Decoding by linear programming'', {\it IEEE Trans. on Information
Theory}, 51(12), pp. 4203 - 4215, December 2005.


\bibitem{D} D. Donoho and J. Tanner, ``Thresholds for the Recovery of Sparse Solutions via L1 Minimization",\emph{ Proceedings of the Conference on Information Sciences and Systems}, March 2006.


\bibitem{Baraniuk Manifolds} R. G. Baraniuk and M. B. Wakin
``Random Projections of Smooth Manifolds '', {\it Journal of Foundations of Computational Mathematics},     Volume 9, No.1, Feb. 2009,

%
%
\bibitem{Donoho positive}
D. Donoho and J. Tanner, ``Sparse nonnegative solutions of underdetermined linear equations by linear programming'' Proc.  National Academy of Sciences, 102(27), pp.9446-9451, 2005.
%
%
%
%
%
%


\bibitem{Needell} D. Needell, ``Noisy signal recovery via iterative reweighted L1-minimization"
     Proc. Asilomar Conf. on Signals, Systems, and Computers, Pacific Grove, CA Nov. 2009.


\bibitem{isitweighted}   A. Khajehnejad, W. Xu, A. Avestimehr, Babak Hassibi, ``Weighted $\ell_1$ minimization for Sparse Recovery with Prior Information",
accepted to \emph{the International Symposium on Information Theory
2009}, available online http://arxiv.org/abs/0901.2912. Complete jounrnal manuscript to be submitted.

\bibitem{isitrobust} W. Xu and B.  Hassibi, ``On Sharp Performance Bounds for Robust Sparse
Signal Recoveries", accepted to \emph{the International Symposium on
Information Theory 2009}.


\bibitem{CWB07} E. J. Cand\'{e}s, M. B. Wakin, and S. Boyd, ``Enhancing Sparsity by Reweighted l1
Minimization", {\em Journal of Fourier Analysis and Applications},
14(5), pp.~877-905, special issue on sparsity, December 2008.

\bibitem{Allerton reweighted l_1}
B. Hassibi, A. Khajehnejad, W. Xu, S. Avestimehr, ''Breaking the $L_1$ Recovery
Thresholds with Reweighted $L_1$ Optimization,", in proceedings of Allerton 2009.

\bibitem{Rice} Compressive sesing online resources at Rice university, http://www.dsp.ece.rice.edu/cs

\bibitem{sotware link}
http://www.its.caltech.edu/$\sim$amin/weighted\_l1\_codes/

\end{thebibliography}
%

\appendix{}
\section{Computation of $\delta_c$ Threshold}
\label{App:delta_c}
\begin{center}A. Computation of $\delta_c$ Threshold\end{center}
The following formulas for $\delta_c(\gamma_1,\gamma_2,f_1,f2_,\omega)$ are given in \cite{isitweighted}.
\begin{align}
\nonumber \delta_c = &\min\{\delta~|~\psi_{com}(\tau_1,\tau_2)-\psi_{int}(\tau_1,\tau_2)-\psi_{ext}(\tau_1,\tau_2)<0~ \\
\nonumber &\forall ~0\leq \tau_1\leq \gamma_1(1-f_1), 0\leq \tau_2\leq \gamma_2(1-f_2),\\
\nonumber &\tau_1+\tau_2 > \delta-\gamma_1f_1-\gamma_2f_2 \}
\end{align}
\noindent where $\psi_{com}$, $\psi_{int}$ and $\psi_{ext}$ are obtained as follows. Define $g(x)=\frac{2}{\sqrt{\pi}}e^{-{x^2}}$, $G(x)=\frac{2}{\sqrt{\pi}}\int_{0}^{x}e^{-y^2}dy$ and let $\varphi(.)$ and $\Phi(.)$ be the standard Gaussian pdf and cdf functions respectively.
\begin{align}
\nonumber &\psi_{com}(\tau_1,\tau_2) = (\tau_1+\tau_2 + \gamma_1(1-f_1)H(\frac{\tau_1}{\gamma_1(1-f_1)})\\
&+\gamma_2(1-f_2)H(\frac{\tau_2}{\gamma_2(1-f_2)}) + \gamma_1H(f_1) + \gamma_2H(f_2)) \log{2}
\end{align}
\noindent where $H(\cdot)$ is the Shannon entropy function. Define $c=(\tau_1+\gamma_1f_1)+\omega ^2(\tau_2+\gamma_2f_2)$, $\alpha_1=\gamma_1(1-f_1)-\tau_1$ and $\alpha_2=\gamma_2(1-f_2)-\tau_2$. Let $x_0$ be the unique solution to $x$ of the equation
$2c-\frac{g(x)\alpha_1}{xG(x)}-\frac{\omega g(\omega x)\alpha_2}{xG(\omega x)}=0$. Then
\begin{equation}
\psi_{ext}(\tau_1,\tau_2) = cx_0^2-\alpha_1\log{G(x_0)}-\alpha_2\log{G(\omega x_0)}
\end{equation}

\noindent Let $b=\frac{\tau_1+\omega ^2\tau_2}{\tau_1+\tau_2}$, $\Omega'=\gamma_1f_1+\omega ^2\gamma_2f_2$ and $Q(s)=\frac{\tau_1\varphi(s)}{(\tau_1+\tau_2)\Phi(s)}+\frac{\omega \tau_2\varphi(\omega s)}{(\tau_1+\tau_2)\Phi(\omega s)}$. Define the function $\hat{M}(s)=-\frac{s}{Q(s)}$ and solve for $s$ in $\hat{M}(s)=\frac{\tau_1+\tau_2}{(\tau_1+\tau_2)b+\Omega'}$. Let the unique solution be $s^*$ and set $y=s^*(b-\frac{1}{\hat{M}(s^*)})$. Compute the rate function $\Lambda^*(y)= sy -\frac{\tau_1}{\tau_1+\tau_2}\Lambda_1(s)-\frac{\tau_2}{\tau_1+\tau_2}\Lambda_1(\omega s)$ at the point $s=s^*$, where $\Lambda_1(s) = \frac{s^2}{2} +\log(2\Phi(s))$.
The internal angle exponent is then given by:
\begin{equation}
\psi_{int}(\tau_1,\tau_2) = (\Lambda^*(y)+\frac{\tau_1+\tau_2}{2\Omega'}y^2+\log2)(\tau_1+\tau_2)
\end{equation}

\end{document}